\documentclass[12pt]{article}

\usepackage{latexsym}
\usepackage{amsmath}
\usepackage{amssymb}
\usepackage{amsthm}
\usepackage{amscd}
\usepackage[mathscr]{eucal}
\usepackage{graphicx}
\usepackage{subfigure}
\usepackage{psfrag}
\usepackage{rotating}
\usepackage{appendix}
\usepackage[all]{xy}

\usepackage{color}

\newcommand{\R}{\mathbb{R}}

\newcommand{\prad}{\widehat{p}_\mathrm{rad}}
\newcommand{\ptan}{\widehat{p}_\mathrm{tan}}

\newtheorem{theorem}{Theorem}

\newtheorem{lemma}{Lemma}
\newtheorem{proposition}{Proposition}

\newtheorem{definition}{Definition}

\usepackage{fullpage}
\numberwithin{equation}{section}

\title{Self-gravitating static balls \\ of power-law elastic matter}
 \author
 {Artur~Alho  \\
           {\small Center for Mathematical Analysis, Geometry and Dynamical Systems}  \\
      {\small Instituto Superior T\'ecnico, Universidade de Lisboa}  \\[0.4cm]
       Simone Calogero, Astrid Liljenberg  \\
       {\small Department of Mathematical Sciences}  \\
       {\small Chalmers University of Technology, University of Gothenburg}  \\
       {\small Gothenburg, Sweden} 
       }

\date{}
\begin{document}

\maketitle
\begin{abstract}
We study a class of power-law stored energy functions for spherically symmetric elastic bodies that includes well-known material models, such as the Saint Venant-Kirchhoff, Hadamard, Signorini and John models. We identify a finite subclass of these stored energy functions, which we call Lam\'e type, that depend on no more material parameters than the bulk modulus $\kappa>0$ and the Poisson ratio $-1<\nu\leq1/2$. A general theorem proving the existence of static self-gravitating elastic balls for some power-law materials has been given elsewhere. In this paper numerical evidence is provided that some hypotheses in this theorem are necessary, while others are not. 
\end{abstract}

\section{Introduction}
Spherically symmetric static configurations of self-gravitating matter distributions are described by the equation
\begin{equation}\label{generalstatic}
p_\mathrm{rad}'=\frac{2}{r}(p_\mathrm{tan}-p_\mathrm{rad})-G\rho\frac{m}{r^2},
\end{equation}
in which $p_\mathrm{rad}(r)$ is the radial pressure, $p_\mathrm{tan}(r)$ the tangential pressure and $\rho(r)$ the mass density of the matter distribution, while 
\[
m(r)=4\pi\int_0^r\rho(s)\,s^2\,ds
\]
is the mass enclosed in the ball of radius $r>0$; $G$ is Newton's gravitational constant. 
Solutions of~\eqref{generalstatic}  are of paramount importance in astrophysics, where, depending on the matter model being used, describe stars, galaxies, planets, or other systems, in static equilibrium~\cite{BT,KW,MW16}. The matter model is specified by assigning an equation of state between the Euler state variables $(\rho,p_\mathrm{rad},p_\mathrm{tan})$, e.g., $p_\mathrm{rad}=p_\mathrm{tan}=F(\rho)$ for a barotropic fluid; in the case of kinetic matter models, the Euler state variables are given by integral moments of the one-particle distribution function $f$ in phase-space~\cite{BT}.

In this paper we study Equation~\eqref{generalstatic} for single balls of elastic matter with stored energy function $\widehat{w}:(0,\infty)^2\to\R$. We use the formulation of elasticity theory for spherically symmetric bodies with natural reference state introduced in~\cite{AC18}, see also~\cite{SC}, in which the Euler state variables of elastic balls satisfy the equations of state
\[
\rho(r)=\mathcal{K}\delta(r),\quad p_\mathrm{rad}(r)=\prad(\delta(r),\eta(r)),\quad p_\mathrm{tan}(r)=\ptan(\delta(r),\eta(r)),
\]
where the constitutive function $(\prad,\ptan)$ is given by
\begin{equation}\label{constfun}
\prad(\delta,\eta)=\delta^2\partial_\delta \widehat{w} (\delta,\eta),\quad \ptan(\delta,\eta)=\prad(\delta,\eta)+\frac{3}{2}\delta\eta\partial_\eta\widehat{w}(\delta,\eta).
\end{equation}
Here $\mathcal{K}$ is called the reference density of the spherically symmetric elastic body and
\[
\eta(r)=\frac{m(r)}{\frac{4\pi}{3}\mathcal{K}r^3}.
\]
In terms of the variables $\delta,\eta$, Equation~\eqref{generalstatic} reads
\begin{subequations}\label{staticsystem}
\begin{equation}
\widehat{a}(\delta,\eta)\delta'=\frac{\widehat{b}(\delta,\eta)}{r}(\eta-\delta)- \frac{4\pi G}{3}\mathcal{K}^2r\,\eta\,\delta,\quad \eta(r)=\frac{3}{r^3}\int_0^r\delta(s)s^2\,ds,
\end{equation}
where
\begin{equation}\label{abc}
\widehat{a}(\delta,\eta)=\partial_\delta\prad(\delta,\eta),\quad \widehat{b}(\delta,\eta)=2\frac{\widehat{p}_\mathrm{tan}(\delta,\eta)-\widehat{p}_\mathrm{rad}(\delta,\eta)}{\eta-\delta}+3\partial_\eta\widehat{p}_\mathrm{rad}(\delta,\eta).
\end{equation}
\end{subequations}
Note that $\eta(0)=\delta(0)$ holds for regular solutions of~\eqref{staticsystem} and it can be shown that $\eta(r)>\delta(r)$ holds for $r>0$, see~\cite{SC}. 
The stored energy function $\widehat{w}$ will be assumed to be of the power-law type introduced in~\cite{AC19}. Several important examples of elastic material models belong to this class, e.g., the Saint-Venant Kirchhoff model, the Signorini model, the Hadamard model and the John model, see Section~\ref{exampelsSec}. (The Ogden material model~\cite{Ogden} also belongs to this class, but it is not included in this paper.)  
Elastic matter models have long been used in astrophysics~\cite{Jeans,Lord,Love, Love2} and have important applications to e.g.~describe the deformation of planets~\cite{JKJG,MW16} and neutron stars crusts~\cite{CH}. 
The existence of static Newtonian self-gravitating elastic bodies, without any symmetry assumption, has been studied in~\cite{BS1,CT} using the Lagrangian formulation of elasticity theory. 
The first theorem proving the existence of static self-gravitating multi-body  elastic matter distributions with regular boundaries and arbitrarily large strain has been given in~\cite{AC18} for the Seth model in spherical symmetry and it was later extended to more general elastic models for static self-gravitating balls in~\cite{AC19}. One purpose of this paper is to present numerical evidence showing that some of the assumptions made in~\cite{AC19} are necessary, while other are not, see Section~\ref{numsec}. In the next section we define and discuss some general properties of power-law stored energy functions.
 
\section{Power-law hyperelastic constitutive functions}
Let $\kappa>0$ be the bulk modulus and $-1<\nu\leq 1/2$ be the Poisson ratio of the material. Several stored energy functions found in literature have the form presented in the following definition.
   \begin{definition}\label{powerdef}
   Let $(n_1,n_2,\dots, n_m)\in\mathbb{N}^m$, $m\geq 2$, and $\theta_j,\beta_{ij}\in\R, i=1,\dots,n_j,\ j=1,\dots,m$ be such that  
\begin{itemize}
\item[(i)] $\theta_1<\theta_2<\dots< \theta_m$, $\beta_{1j}<\beta_{2j}<\dots <\beta_{n_jj}$, for all $j=1,\dots,m$; 
\item[(ii)] if $n_j=1$, then $\theta_j\neq 0$ and $\theta_j=\beta_{1j}$;
\item[(iii)] at least one of the numbers $\beta_{ij}$ is different from 0 and $-1$;
\item[(iv)] at least one of the positive integers $n_j$ is greater than 1.
\end{itemize}
Assume that there exists an interval $V\subseteq(-1,1/2 ]$ such that for all $\nu\in V$ the following linear system
\begin{subequations}\label{condsyst}
\begin{align}
&\sum_{j=1}^{m}\sum_{i=1}^{n_j}\alpha_{ij}\theta_j=0,\quad\sum_{j=1}^{m}\sum_{i=1}^{n_j}\alpha_{ij}\theta_j^2=1,\quad  \sum_{j=1}^m\sum_{i=1}^{n_j}\alpha_{ij}\beta_{ij}^2=3\frac{1-\nu}{1+\nu},\label{cond2} \\ 
&\sum_{i=1}^{n_j}\alpha_{ij}(\theta_j-\beta_{ij})=0,\quad j=1,\dots, m\label{cond3}
\end{align}
\end{subequations}
has a solution $\alpha_{ij}\in\R\diagdown\{0\}$, $ i=1,\dots,n_j,\ j=1,\dots,m$.
Then the function  $\widehat{w}:(0,\infty)^2\to\R$ given by
\begin{equation}\label{stored}
\kappa^{-1}\widehat{w}(\delta,\eta)=\sum_{j=1}^m \eta^{\theta_j}\sum_{i=1}^{n_j}\alpha_{ij}\left(\frac{\delta}{\eta}\right)^{\beta_{ij}} +w_0,\quad w_0:=-\sum_{j=1}^m\sum_{i=1}^{n_j}\alpha_{ij}
\end{equation}
 is said to be a type $(n_1,\dots,n_m)$ power-law elastic stored energy function for spherically symmetric bodies with natural reference state.
 \end{definition}
 {\it Remark.} Examples of well-known hyperelastic models with power-law stored energy function are the Saint Venant-Kirchhoff model, the Signorini model, the John model and the Hadamard model, see~\cite{AC19} and Section~\ref{exampelsSec} below. 

The constitutive funtion~\eqref{constfun} for power-law stored energy functions is given by
\begin{align}
&\kappa^{-1}\prad(\delta,\eta)=\sum_{j=1}^{m} \eta^{1+\theta_j} \sum_{i=1}^{n_j}\alpha_{ij}\beta_{ij}\left(\frac{\delta}{\eta}\right)^{1+\beta_{ij}},\label{Prad}\\
& \kappa^{-1}\ptan(\delta,\eta)=\frac{1}{2}\sum_{j=1}^m \eta^{1+\theta_j} \sum_{i=1}^{n_j}\alpha_{ij}(3\theta_j- \beta_{ij})\left(\frac{\delta}{\eta}\right)^{1+\beta_{ij}},\label{Ptan}
\end{align}
while the functions $\widehat{a}(\delta,\eta)$, $\widehat{b}(\delta,\eta)$ in~\eqref{abc} are given by
\begin{subequations}\label{abcpower}
\begin{align}
&\kappa^{-1}\widehat{a}(\delta,\eta)=\sum_{j\in J}\eta^{\theta_j}\sum_{i\in I_j}\alpha_{ij}\beta_{ij}(1+\beta_{ij})\left(\frac\delta\eta\right)^{\beta_{ij}},\\
&\kappa^{-1}\widehat{b}(\delta,\eta)=3\sum_{j\in J}\eta^{\theta_j}\sum_{i\in I_j}\alpha_{ij}(\theta_j-\beta_{ij})\left[\frac{(\delta/\eta)^{1+\beta_{ij}}-1}{1-(\delta/\eta)}+\beta_{ij}\left(\frac{\delta}{\eta}\right)^{1+\beta_{ij}}\right],  
\end{align}
\end{subequations}
where we introduced the sets 
\[
I_j:=\{i\in\{1,\dots,n_j\}:\beta_{ij}\neq -1, \beta_{ij}\neq 0\},\quad J=\{j\in\{1,\dots,m\}: I_j\neq\varnothing\}.
\]
By condition (iii) on $\beta_{ij}$ in Definition~\ref{powerdef}, $I_j$ is not empty for at least one $j\in\{1,\dots,m\}$.

For a stored energy function of the form~\eqref{stored}, the definition of $w_0$ is equivalent to the normalization condition $\widehat{w}(1,1)=0$.
Equations~\eqref{cond2} are equivalent to the following compatibility equations with linear elasticity:
\begin{subequations}\label{lincomp}
\begin{equation}\label{lincomp1}
\kappa^{-1}\partial_\delta\prad(1,1)=\frac{3(1-\nu)}{1+\nu},\quad \kappa^{-1}\partial_\eta\prad(1,1)=-\frac{2(1-2\nu)}{1+\nu},
\end{equation}
\begin{equation}\label{lincomp2}
\kappa^{-1}\partial_\delta\ptan(1,1)=\frac{3\nu}{1+\nu},\quad\kappa^{-1}\partial_\eta\ptan(1,1)=\frac{1-2\nu}{1+\nu},
\end{equation}
\end{subequations}
and the natural reference state condition
\begin{equation}\label{naturalstate}
\prad(1,1)=\ptan(1,1)=0,
\end{equation}
while~\eqref{cond3} is equivalent to isotropic condition
\begin{equation}\label{regularcentercond}
\widehat{p}_\mathrm{rad}(\delta,\delta)=\widehat{p}_\mathrm{tan}(\delta,\delta).
\end{equation}
We refer to~\cite{AC18, SC} for a more detailed discussion on the conditions~\eqref{lincomp}-\eqref{regularcentercond}. 

{\it Remark.} Condition (i) only affects the order in which the factors $\eta^{1+\theta_j}(\delta/\eta)^{1+\beta_{ij}}$ appear in the stored energy function (lexicographic order); condition (ii) is required for consistency with~\eqref{cond3}, while condition (iii) is imposed to ensure that the constitutive function for the radial pressure is not independent of $\delta$, see~\eqref{Prad}. Condition (iv) excludes the power-law types $(1,1,1,\dots, 1)$ from Definition~\ref{powerdef}. These stored energy functions correspond to barotropic fluid models and will be discussed in a separate section, see the last example in Section~\ref{exampelsSec}.


{\it Remark. }
In~\cite{AC19} we used the Lam\'e material parameters $\lambda,\mu$ instead of $\kappa,\nu$. The relation between these two sets of parameters is 
\[
\kappa=\lambda+\frac{2\mu}{3},\quad\nu=\frac{\lambda}{2(\lambda+\mu)}.
\]

{\it Remark.} 
The assumption that the system~\eqref{condsyst} must have non-zero solutions for all $\nu$ in some interval $V\subseteq (-1,1/2 ]$ is required to exclude the possibility that the stored energy function~\eqref{stored} is defined only for isolated values of the Poisson ratio. The latter seems rather artificial and would result in unpleasant technical complications.

{\it Remark.} For some applications, e.g., to study the homologous motion of (self-gravitating) elastic balls~\cite{SC}, it is necessary to consider elastic stored energy functions which do not satisfy the natural state condition~\eqref{naturalstate}.

If there are no values of $\theta_j$, $\beta_{ij}$  for which the system~\eqref{condsyst} admits solutions $\alpha_{ij}\neq 0$ for all $\kappa>0, \nu\in V$, for some interval $V\subseteq(-1,1/2]$, then the corresponding power-law stored energy function is inadmissible. To this regard we have the following simple lemma.
\begin{lemma}\label{inadstored}
There are no power-law stored energy functions of type (1,2), (2,1) or (2,2).
\end{lemma} 
\begin{proof}
When $m=2$ and $(n_1,n_2)=(1,2)$, the system~\eqref{condsyst} is given by
\[
\mathrm{(*)}\left\{
\begin{array}{l}
\alpha_{11}\theta_1+\alpha_{12}\theta_2+\alpha_{22}\theta_2=0,\\
\alpha_{11}\theta_1^2+\alpha_{12}\theta_2^2+\alpha_{22}\theta_2^2=1,\\
\alpha_{12}(\theta_2-\beta_{12})+\alpha_{22}(\theta_2-\beta_{22})=0,\\
\alpha_{11}\theta_{1}^2+\alpha_{12}\beta_{12}^2+\alpha_{22}\beta_{22}^2=3\frac{1-\nu}{1+\nu}.
\end{array}
\right.
\]
The system consisting of the first three equations has no solution when $\theta_1$ or $\theta_2$ is zero, while for $(\theta_1,\theta_2)\neq(0,0)$ has the unique solution
\[
\alpha_{11}=\frac{1}{\theta_1(\theta_1-\theta_2)},\quad \alpha_{12}=\frac{\theta_2-\beta_{22}}{\theta_2(\theta_1-\theta_2)(\beta_{22}-\beta_{12})},\quad \alpha_{22}=\frac{\beta_{12}-\theta_2}{\theta_2(\theta_1-\theta_2)(\beta_{22}-\beta_{12})},
\]
which is non-zero provided $\theta_2\neq \beta_{i2}$, $i=1,2$. Replacing in the fourth equation we obtain that the system $\mathrm{(*)}$ has no solutions, except possibly for an isolated value of the Poisson ratio and thus this model is not a type (1,2) power-law stored energy function. Similarly one proves that the types (2,1) and (2,2) are inadmissible.
\end{proof}

%

{\it Remark.} From the simple proof of the previous lemma it is clear that a power-law type is inadmissible if the subsystem of~\eqref{condsyst} consisting of~\eqref{cond3} and the first two equations in~\eqref{cond2} has a unique solution. The special result included in Lemma~\ref{inadstored} will be used in Lemma~\ref{lametheo}.

\subsection{Lam\'e type power-law stored energy functions}
We shall say that a power-law stored energy function is of Lam\'e type if the coefficients $\alpha_{ij}$ in~\eqref{stored} are uniquely determined by the exponents $\theta_j, \beta_{ij}$ and the Poisson ratio $\nu$ through the system~\eqref{condsyst}. Power-law stored energies which are not Lam\'e type contain additional parameters besides the bulk modulus $\kappa$ and the Poisson ratio $\nu$. We remark that these additional parameters are not genuine material constants, as they depend on having assumed a specific type of stored energy function (while $\kappa$ and $\nu$ only depend on the postulate that {\it all} materials obey linear elasticity for very small strain). 
As shown in the following lemma, there are only a few Lam\'e types power-law materials.
\begin{lemma}\label{lametheo}
The only possible elastic stored energy functions of Lam\'e type are the following:
\begin{subequations}\label{lametypes}
\begin{align}
&m=2:  (1,3),  (2,3)\quad\text{and permutations},\\
&m=3:  (1,1,2), (1,2,2), (2,2,2) \quad\text{and permutations.}
\end{align}
\end{subequations}
\end{lemma}
\begin{proof}
Let $k$ be the number of $n_j= 1$. As the system~\eqref{condsyst} consists of $m+3-k$ equations and there are $\sum_{j=1}^m n_j$ coefficients $\alpha_{ij}$, then a necessary condition for a power-law material to be of Lam\'e type is that $m+3-k\geq\sum_{j=1}^m n_j$. Using 
\[
\sum_{j=1}^m n_j=k+\sum_{n_j\geq 2}n_j\geq k+2(m-k)
\]
we find that $m\leq 3$. Hence the only possible Lam\'e types are (1,2), (1,3), (2,2) and (2,3) for $m=2$,  (1,1,2), (1,2,2), (2,2,2) for $m=3$, and permutations thereof. Having shown in Lemma~\ref{inadstored} that the types (1,2), (2,1) and (2,2) are inadmissible, the proof is completed.
\end{proof} 

Upon studying each of the types~\eqref{lametypes} separately, one can easily show that they are all Lam\'e types except for some some special values of the exponents $\theta_1,\dots,\theta_m$. For instance, the most general type $(2,2,2)$ power-law stored energy function is
\begin{align}
&\kappa^{-1}\widehat{w}(\delta,\eta)=\eta^{\theta_1}\left(\alpha_{11}\left(\frac{\delta}{\eta}\right)^{\beta_{11}}+\alpha_{21}\left(\frac{\delta}{\eta}\right)^{\beta_{21}}\right)+\eta^{\theta_2}\left(\alpha_{12}\left(\frac{\delta}{\eta}\right)^{\beta_{12}}+\alpha_{22}\left(\frac{\delta}{\eta}\right)^{\beta_{22}}\right)\nonumber\\
&\quad+\eta^{\theta_3}\left(\alpha_{13}\left(\frac{\delta}{\eta}\right)^{\beta_{13}}+\alpha_{23}\left(\frac{\delta}{\eta}\right)^{\beta_{23}}\right)-(\alpha_{11}+\alpha_{21}+\alpha_{12}+\alpha_{22}+\alpha_{13}+\alpha_{23}).\label{general12}
\end{align}
The system~\eqref{condsyst} consists of 6 equations on $\alpha_{11},\alpha_{21},\alpha_{12},\alpha_{22},\alpha_{13},\alpha_{23}$. 
Defining
\begin{align*}
A_{(2,2,2)}&=\theta_1\theta_2(\theta_1-\theta_2)(\beta_{13}-\theta_3)(\beta_{23}-\theta_3)+\theta_2\theta_3(\theta_2-\theta_3)(\beta_{11}-\theta_1)(\beta_{21}-\theta_1)\\
&\quad+\theta_1\theta_3(\theta_3-\theta_1)(\beta_{12}-\theta_2)(\beta_{22}-\theta_2),
\end{align*}
the system~\eqref{condsyst} for type (2,2,2) power-law stored energy functions has a unique solution if and only if $A_{(2,2,2)}\neq 0$, namely
\begin{align*}
&\alpha_{11}=\frac{\theta_1-\beta_{21}}{A(\beta_{11}-\beta_{21})}[(\theta_2(\beta_{23}-\theta_3)(\beta_{13}-\theta_3)-\theta_{3}(\beta_{22}-\theta_2)(\beta_{12}-\theta_2)+2\theta_2\theta_3(\theta_3-\theta_2)\frac{1-2\nu}{1+\nu}]\\
&\alpha_{12}=\frac{\beta_{22}-\theta_2}{A(\beta_{12}-\beta_{22})}[(\theta_1(\beta_{23}-\theta_3)(\beta_{13}-\theta_3)-\theta_{3}(\beta_{21}-\theta_1)(\beta_{11}-\theta_1))+2\theta_1\theta_3(\theta_3-\theta_1)\frac{1-2\nu}{1+\nu}]\\
&\alpha_{13}=\frac{\theta_3-\beta_{23}}{A(\beta_{13}-\beta_{23})}[(\theta_1(\beta_{22}-\theta_2)(\beta_{12}-\theta_2)-\theta_{2}(\beta_{21}-\theta_1)(\beta_{11}-\theta_1))+2\theta_1\theta_2(\theta_2-\theta_1)\frac{1-2\nu}{1+\nu}]\\
&\alpha_{21}=-\frac{\beta_{11}-\theta_1}{\beta_{21}-\theta_1}\alpha_{11},\quad\alpha_{22}=-\frac{\beta_{12}-\theta_2}{\beta_{22}-\theta_2}\alpha_{12},\quad\alpha_{23}=-\frac{\beta_{13}-\theta_3}{\beta_{23}-\theta_3}\alpha_{13}.
\end{align*}
From these expressions it is clear that the conditions on the exponents $\theta_j$, $\beta_{ij}$ for the existence of the Lam\'e type (2,2,2) power-law stored energy function are
\[
(\theta_1,\theta_2,\theta_3)\neq (0,0,0),\quad \theta_{j}\neq \beta_{ij},\quad i=1,2,\quad j=1,2,3.
\] 
Similar conditions can be found for the other types in~\eqref{lametypes}. Moreover all Lam\'e types can be derived from the types (2,2,2), (1,3) and (3,1) in the limits given in Figure~\ref{lamefig}.
\begin{figure}
\begin{center}
\begin{minipage}[t]{0.4\textwidth}
\[
 \xymatrix{
         & & (1,2,1)\\
        & (2,2,1) \ar[ur]^{\beta_{11}=\theta_1} \ar[r]^{\beta_{12}=\theta_2} & (2,1,1)\\
        (2,2,2) \ar[ur]^{\beta_{13}=\theta_3}  \ar[dr]_{\beta_{11}=\theta_1} \ar[r]^{\beta_{12}=\theta_2} &(2,1,2)\ar[ur]^{\beta_{13}=\theta_3} \ar[r]^{\beta_{11}=\theta_1} &(1,1,2)\\
        &(1,2,2) \ar[ur]^{\beta_{12}=\theta_2} \ar[r]_{\beta_{13}=\theta_3} &(1,2,1)
         }
         \]
         \end{minipage}\qquad
         \begin{minipage}[t]{0.4\textwidth}
         \[
         \xymatrix{
         (2,3) \ar[r]^{\beta_{11}=\theta_1}& (1,3) \ar[dr]^{\nu=1/2}&\\
         & &(1,1)\\
         (3,2)\ar[r]^{\beta_{12}=\theta_2} & (3,1)\ar[ur]_{\nu=1/2}&
         }
         \]
      {\it Remark.} The same limits hold by replacing $\beta_{1j}$ with $\beta_{2j}$.
        \end{minipage} 
        \end{center}
        \caption{Relation between the Lam\'e type power-law stored energy functions. 
        }\label{lamefig}
        \end{figure}
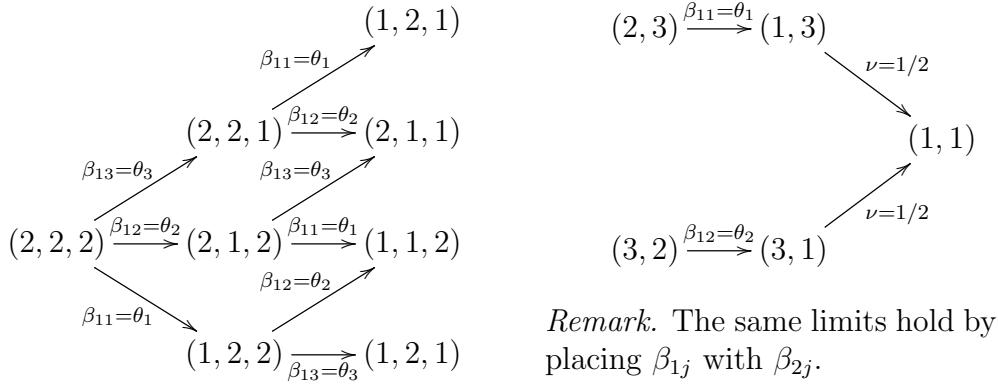

\subsection{Examples}\label{exampelsSec}
In this section we give some examples of power-law stored energy functions for spherically symmetric bodies; the original Lagrangian form of these models without symmetry assumptions can be found in~\cite{Ciarlet, John, MH, Sig43}. 
\subsubsection*{Saint Venant-Kirchhoff model}
Saint Venant-Kirchhoff (SVK) materials have the following Lam\'e type $(3,2)$ power-law stored energy function
\begin{align}\label{SVK}
\kappa^{-1}\widehat{w}_\mathrm{SVK}(\delta,\eta)&=\eta^{-4/3}\left(\frac{3(1-\nu)}{8(1+\nu)}\left(\frac{\delta}{\eta}\right)^{-4}+\frac{3\nu}{2(1+\nu)}\left(\frac{\delta}{\eta}\right)^{-2}+\frac{3}{4(1+\nu)}\right)\nonumber\\
&\quad+\eta^{-2/3}\left(-\frac{3}{4}\left(\frac{\delta}{\eta}\right)^{-2}-\frac{3}{2}\right)+\frac{9}{8}.
\end{align}

\subsubsection*{(Quasi linear) Signorini model}
Signorini materials have the following non-Lam\'e type power-law stored energy function
\begin{align}\label{Sig}
\kappa^{-1}\widehat{w}_\mathrm{Sig}(\delta,\eta) = &\frac{3(5+8\nu)-\tau(1+\nu)}{16(1+\nu)}\eta^{-1}\left(\frac{\delta}{\eta}\right)^{-1}\nonumber\\
&+\eta^{-\frac{1}{3}}\left(\frac{\tau(1+\nu)-3(1+4\nu)}{4(1+\nu)}\left(\left(\frac{\delta}{\eta}\right)^{-1}+\frac{1}{2}\left(\frac{\delta}{ \eta}\right)\right)\right) \nonumber \\
&+\eta^{\frac{1}{3}} \left(\frac{3}{4(1+\nu)}\left(\frac{\delta}{\eta}\right)^{-1}+\frac{3+\tau(1+\nu)}{4(1+\nu)}\left(\frac{\delta}{\eta}\right)+\frac{3-\tau(1+\nu)}{16(1+\nu)}\left(\frac{\delta}{\eta}\right)^3\right)\nonumber\\
&-\frac{3(1-2\nu)+\tau(1+\nu)}{2(1+\nu)},
\end{align}
where $\tau$ is a (dimensionless) constant. 
Except for some particular values of $\tau$, this stored energy function is of type (1,2,3). 
The case $\tau=0$ is known as quasi linear Signorini model. We shall restrict to the latter case in the rest of the paper. 

\subsubsection*{(Quasi linear) John model}
John materials have the following non-Lam\'e type power-law stored energy function
\begin{align}\label{John}
\kappa^{-1}\widehat{w}_\mathrm{John}(\delta,\eta) = &-\frac{\epsilon(1+4\nu)+3(1-2\nu)}{(1+4\nu)}\eta^{-1}\left(\frac{\delta}{ \eta}\right)^{-1}\nonumber\\
&+\eta^{-\frac{2}{3}}\left(\frac{3}{2(1+4\nu)}\left(\frac{\delta}{\eta}\right)^{-2}+2\frac{3+\epsilon(1+4\nu)}{1+4\nu}\left(\frac{\delta}{ \eta}\right)^{-1}+\frac{6+\epsilon(1+4\nu)}{1+4\nu}\right) \nonumber \\
&+\eta^{-\frac{1}{3}} \left(-\frac{6(1+\nu)+\epsilon(1+4\nu)}{1+4\nu}\left(\frac{\delta}{\eta}\right)^{-1}-2\frac{6(1+\nu)+\epsilon(1+4\nu)}{1+4\nu}\right)\nonumber\\
&+\frac{2\epsilon(1+4\nu) +3(5+8\nu)}{2(1+4\nu)},
\end{align}
where $\epsilon$ is a (dimensionless) constant. 
Except for some particular values of $\epsilon$, this stored energy function is of type (1,3,2).  
From now on we restrict to the case $\epsilon=0$, which we called the quasi-linear John model. 

{\it Remark.} The John model discussed in this section is actually just a special case of a larger family of stored energy functions, called harmonic, introduced by Fritz John in~\cite{John}. 

\subsubsection*{Hadamard model}
Hadamard materials are hyperelastic materials with the following Lam\'e type (2,1,2) stored energy function
\begin{align}\label{Had}
\kappa^{-1}\widehat{w}_{\mathrm{Had}}(\delta,\eta)&=\eta^{-4/3}\left(\frac{3}{2(1+\nu)}\left(\frac{\delta}{\eta}\right)^{-2}+\frac{3}{4(1+\nu)}\right)-\eta^{-1}\left(\frac{3(1-\nu)}{1+\nu}\left(\frac{\delta}{\eta}\right)^{-1}\right)\nonumber\\
&\quad+\eta^{-2/3}\left(-\frac{3\nu}{2(1+\nu)} \left(\frac{\delta}{\eta}\right)^{-2}-\frac{3\nu}{1+\nu}\right)+\frac{3(1+2\nu)}{4(1+\nu)},\quad \nu> 0.
\end{align}

{\it Remark.} Hadamard materials are defined up to an additive term $h(\delta)$, which was chosen $h(\delta)\sim\delta^{-1}$ in~\eqref{Had}.
\subsubsection*{Affine models}
In this section we derive the power-law stored energy functions for which the Equation~\eqref{staticsystem} satisfied by static self-gravitating bodies admits self-similar type solutions, i.e., solutions of the form $\delta(r)=c r^\alpha$, for some $c,\alpha\in\R$. Positivity of the mass requires $\alpha>-3$. Moreover $\eta(r)=\frac{3}{3+\alpha}\delta(r)$. Substituting in~\eqref{staticsystem}, we obtain the following equation on the constitutive function:
\begin{equation}\label{conditionselfsim}
\widehat{a}\left(\delta,\frac{3}{3+\alpha}\delta\right)\alpha-\widehat{b}\left(\delta,\frac{3}{3+\alpha}\delta\right)\frac{\alpha}{\alpha+3}+\frac{4\pi}{3}\mathcal{K}^2\frac{3}{3+\alpha}c^{-2/\alpha}\delta^{1+2/\alpha}=0.
\end{equation}
A stored energy function that satisfies~\eqref{conditionselfsim} will be called affine.
\begin{proposition}\label{propsim}
A necessary condition for a power-law material to satisfy~\eqref{conditionselfsim} for all $\delta>0$ is that 
\begin{equation}\label{ass}
\text{there exists a unique $q\in\{1,\dots,m\}$ such that $I_q\neq\varnothing$.}
\end{equation}
When~\eqref{ass} holds, 
~\eqref{conditionselfsim} is satisfied with $\alpha>-3$, $c>0$ and for all $\delta>0$ if and only if $\theta:=\theta_q\notin[1/3,1]$,
\[
C(\theta):=\frac{2}{\theta-1}\left[\frac{1-\theta}{1-3\theta}\,\widehat{b}\left(1,\frac{3-3\theta}{1-3\theta}\right)-\widehat{a}\left(1,\frac{3-3\theta}{1-3\theta}\right)\right]>0
\]
and
\begin{equation}\label{alphac}
\alpha=\frac{2}{\theta-1}, \quad c=\left(\frac{3}{4\pi\mathcal{K}^2}\right)^{\frac{1}{1-\theta}}\frac{1-3\theta}{3-3\theta}C(\theta)^{\frac{1}{1-\theta}}.
\end{equation}
\end{proposition}
\begin{proof}
If $I_j\neq \varnothing$ for more than one value of $j\in\{1,\dots, m\}$, the left hand side of~\eqref{conditionselfsim} would contain two different powers of $\delta$ and thus~\eqref{conditionselfsim} cannot be verified for all $\delta>0$. If $I_j\neq\varnothing$ only for a unique $j=q\in\{1,\dots,m\}$ and setting $\theta=\theta_q$ we obtain
\begin{align*}
&\widehat{a}\left(\delta,\frac{3}{3+\alpha}\delta\right)=\widehat{a}\left(1,\frac{3}{3+\alpha}\right)\left(\frac{3}{3+\alpha}\right)^{\theta}\delta^{\theta},\\
&\widehat{b}\left(\delta,\frac{3}{3+\alpha}\delta\right)=\widehat{b}\left(1,\frac{3}{3+\alpha}\right)\left(\frac{3}{3+\alpha}\right)^{\theta}\delta^{\theta},
\end{align*}
hence a necessary condition for~\eqref{conditionselfsim} to hold for all $\delta>0$ is that $\theta=1+2/\alpha$, which gives the formula for $\alpha$ in~\eqref{alphac} as well as the condition $\theta\notin[1/3,1]$. Replacing in~\eqref{conditionselfsim}  we find that~\eqref{conditionselfsim} holds if and only if $c$ is given as in~\eqref{alphac}, hence $C(\theta)>0$ must hold in order that $c>0$. 
\end{proof}
The only example considered so far that satisfies the assumption~\eqref{ass} is the John model, in which case the self similar solution becomes the one found in~\cite{AC19}. Another interesting example of affine power-law stored energy function is the following Lam\'e type:
\begin{align}
&\kappa^{-1}\widehat{w}(\delta,\eta)=\frac{1}{1+\theta}\delta^{-1}-\frac{1}{\theta}\nonumber\\
&+\eta^{\theta}\Big(\frac{3(1-\nu)}{\beta(\beta-1)(1+\nu)}(\delta/\eta)^{\beta-1}+\Big(\frac{3(1-\nu)}{\beta(1+\nu)}-\frac{1}{1+\theta}\Big)(\delta/\eta)^{-1}+\frac{1}{\theta}-\frac{3(1-\nu)}{(\beta-1)(1+\nu)}\Big),\label{newexample}
\end{align}
which is of type (3,1) for $\theta<-1$ and of type (1,3) for $\theta>-1$.
Applying the result of the Proposition~\ref{propsim} to the stored energy function~\eqref{newexample}  we obtain the self-similar solution of~\eqref{staticsystem} given by $\delta(r)=cr^\alpha$, where $\alpha,c$ are given by~\eqref{alphac} and
\[
C(\theta)=\frac{1-\nu}{1+\nu}\frac{18\left(\frac{3-3\theta}{1-3\theta}\right)^\theta(5-2\theta)}{(1-3\theta)^2}>0\quad \text{if and only if }\quad \theta\in (-\infty,1/3)\cup(1,5/2).
\]

{\it Remark.} The stored energy function~\eqref{newexample} is a special case of the class of polytropic stored energy functions introduced in~\cite{SC}.
\subsubsection*{Barotropic fluids}
Power-law stored energy functions of type ($1,1,\dots,1$) correspond to barotropic fluids. The most general stored energy function in this case has the form
\[
\kappa^{-1}\widehat{w}_\mathrm{fluid}^{(m)}(\delta)=\alpha_{11}\delta^{\theta_1}+\alpha_{12}\delta^{\theta_2}+\dots+\alpha_{m1}\delta^{\theta_m}-(\alpha_{11}+\alpha_{12}+\dots+\alpha_{1m}),
\]
where $\theta_1<\theta_2<\dots<\theta_m$ are all different from zero and at least one is different from $-1$.
The system~\eqref{condsyst} on the coefficients $\alpha_{ij}$ reduces to
\begin{subequations}\label{fluidsystem}
\begin{align}
&\alpha_{11}\theta_1+\alpha_{12}\theta_2+\dots \alpha_{1m}\theta_m=0\\
&\alpha_{11}\theta_1^2+\alpha_{12}\theta_2^2+\dots \alpha_{1m}\theta_m^2=1\\
&\alpha_{11}\theta_1^2+\alpha_{12}\theta_2^2+\dots \alpha_{1m}\theta_m^2=3(1-\nu)/(1+\nu).
\end{align}
\end{subequations}
We see that $\nu=1/2$ must hold for the system~\eqref{fluidsystem} to admit solutions. For $m>3$ (and $\nu=1/2$) the system~\eqref{fluidsystem} has infinitely many solutions, while for $m=2$ the system~\eqref{fluidsystem} admits the unique solution
\[
\alpha_{11}=\frac{1}{\theta_1(\theta_1-\theta_2)},\quad \alpha_{12}=-\frac{1}{\theta_2(\theta_1-\theta_2)}.
\]
Thus the type (1,1) is the only Lam\'e type power-law fluid stored energy function. The constitutive functions for the principal pressures of these materials are
\begin{equation}\label{p11}
\prad^{(1,1)}(\delta)=\ptan^{(1,1)}(\delta)=\kappa\frac{\delta(\delta^{\theta_2}-\delta^{\theta_1})}{\theta_2-\theta_1},\quad \theta_2>\theta_1.
\end{equation}
For $\theta_2=\gamma-1$ and $\theta_1=-1$,~\eqref{p11} becomes the constitutive function of polytropic fluids with polytropic exponent $\gamma$, see~\cite{SC}.


\section{Numerical results}\label{numsec}
The purpose of this final section is to investigate numerically whether some of the assumptions made in~\cite{AC19} to prove the existence of static self-gravitating elastic balls are necessary or not. The results concern the Saint Venant-Kirchhoff, John, Hadamard and Signorini model, each discussed in a separate subsection. For each of these models, there exists (a necessarily unique) $\Delta_\flat\in (1,\infty]$ such that $\partial_\delta \prad(\delta,\delta)>0$ for $0<\delta<\Delta_\flat$ and if $\Delta_\flat<\infty$ then $\partial_\delta\prad(\Delta_\flat,\Delta_\flat)=0$; in particular, if $\Delta_\flat<\infty$ and $\delta_c=\delta(0)>\Delta_\flat$, the inequality $\partial_\delta\prad(\delta,\eta)>0$ is violated at the center. As shown in~\cite{SC}, this inequality corresponds to the hyperbolicity condition for the system of equations describing the motion of spherically symmetric elastic balls in Eulerian variables. 

For more numerical results on static self-gravitating solutions for the models in this section, e.g., the mass-radius diagram and the existence of multi-body distributions, we refer to~\cite{Astrid}. 
\subsection{Saint Venant-Kirchhoff materials}\label{sec:svk}
The Saint Venant-Kirchhoff material model is hyperelastic with stored energy function~\eqref{SVK},
which yields the principal pressures
\begin{align*}
    \kappa^{-1}\prad(\delta,\eta) &= \eta^{-1/3}\left(-\frac{3}{2}\frac{1-\nu}{1+\nu}\left(\frac{\delta}{\eta}\right)^{-3}-\frac{3\nu}{1+\nu}\left(\frac{\delta}{\eta}\right)^{-1}\right)+\frac{3}{2}\eta^{1/3}\left(\frac{\delta}{\eta}\right)^{-1}, \\
    \kappa^{-1}\ptan(\delta,\eta) &= \eta^{-1/3}\left(-\frac{3}{2}\frac{\nu}{1+\nu}\left(\frac{\delta}{\eta}\right)^{-1}-\frac{3}{2(1+\nu)}\left(\frac{\delta}{\eta}\right)\right)+\frac{3}{2}\eta^{1/3}\left(\frac{\delta}{\eta}\right).
\end{align*}
At the center
\[
    \kappa^{-1}\prad(\delta_c,\delta_c) = \kappa^{-1}\ptan(\delta_c,\delta_c) = \frac{3 \left(\delta _c^{2/3}-1\right)}{2 \delta _c^{1/3}}, \quad
    \kappa^{-1}\partial_\delta \prad(\delta_c,\delta_c) = \frac{3 \left(3-\nu-(1+\nu) \delta _c^{2/3}\right)}{2 (1+\nu ) \delta _c^{4/3}}.
\]
The pressures are positive at the center if and only if $\delta_c >1$, and the constant $\Delta_\flat$ is given by
\begin{equation*}
    \Delta_\flat = \left( \frac{3-\nu}{1+\nu} \right)^{3/2}.
\end{equation*}
The following theorem was proved in~\cite{AC19}:
\begin{theorem}\label{svktheo}
	When the elastic material is given by the Saint Venant-Kirchhoff model, the condition $\delta_c:=\rho_c/ \mathcal{K}>1$ is necessary for the existence of regular static self-gravitating balls. When $1<\delta_c<\Delta_\flat$ 
	there exists a unique strongly regular static self-gravitating ball with central density $\rho(0)=\rho_c$. 
\end{theorem}
We are interested in the following question: {\it 
Do static self-gravitating elastic balls exist in the Saint Venant-Kirchhoff material model when the strict hyperbolicity condition is violated at the center?}

We have found numerical evidence suggesting that regardless of the value of $\nu\in (-1,1/2]$, finite radius solutions cannot be constructed when $\delta_c \geq \Delta_\flat$. The density and pressures blow up almost immediately, see Figure \ref{fig:SVK}. Note that this kind of blow-up is not possible in the fluid case, since the pressure and the mass density of static self-gravitating fluids are decreasing functions of the radius.



\begin{figure}
    \begin{center}
    \subfigure[$\delta_c = 0.99\Delta_\flat$]{
        \includegraphics[width=0.45\textwidth, trim=0.2cm 0 0.3cm 0,clip]{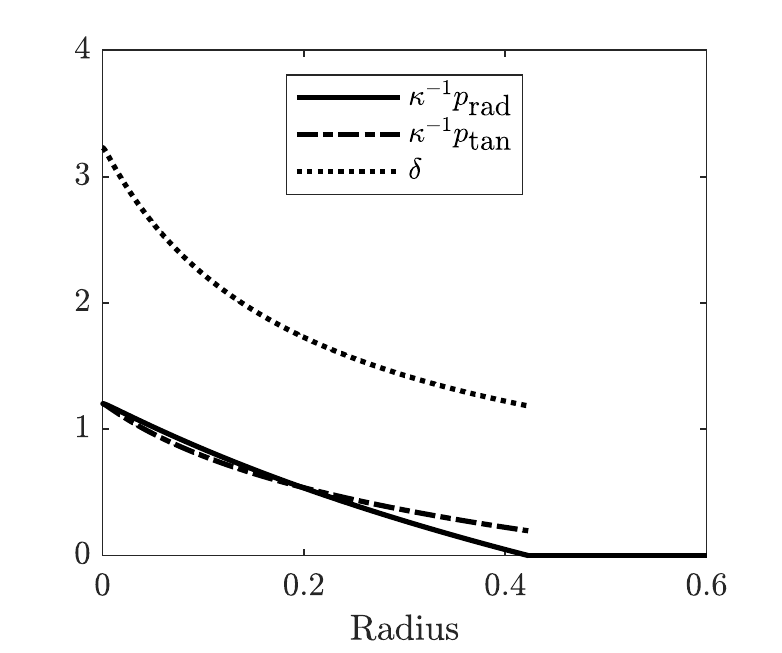}}\quad
        \subfigure[$\delta_c = 1.01\Delta_\flat$]{
        \includegraphics[width=0.45\textwidth, trim=0.2cm 0 0.3cm 0,clip]{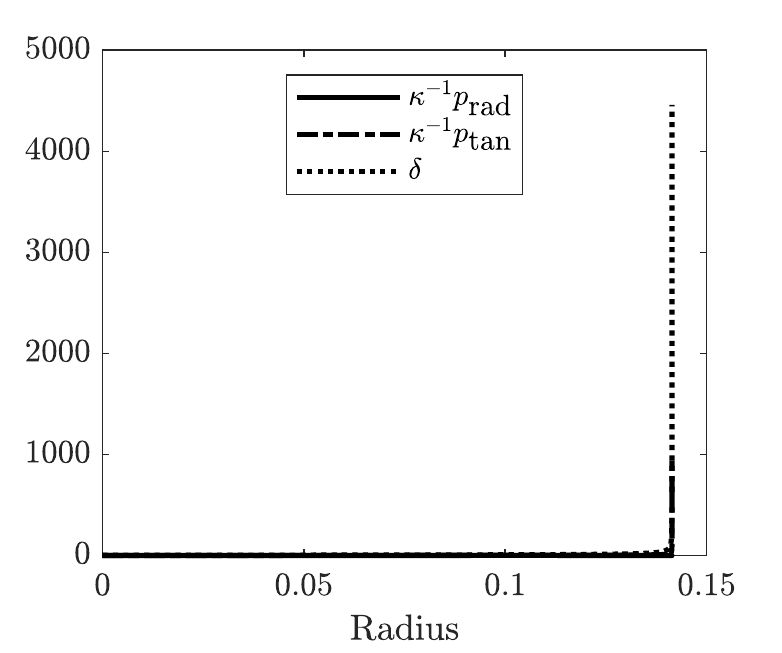}}
    \caption{Elastic balls constructed in the Saint Venant-Kirchoff material model with Poisson ratio $\nu = 0.25$ for center datum close to $\Delta_\flat$ ($\approx 3.26$).}
    \label{fig:SVK}
    \end{center}
\end{figure}



\subsection{Quasi-linear John materials}\label{sec:john}
Quasi-linear John materials are hyperelastic with stored energy function~\eqref{John}$_{\epsilon=0}$, 
which yields the principal pressures
\begin{align*}
&\kappa^{-1}\prad(\delta,\eta)=-\frac{3}{1+4\nu}\eta^{1/3}\left(\left(\frac{\delta}{\eta}\right)^{-1}+2\right)+6\frac{1+\nu}{1+4\nu}\eta^{2/3}+\frac{3(1-2\nu)}{1+4\nu},\\
&\kappa^{-1}\ptan(\delta,\eta)=-\frac{3}{1+4\nu}\eta^{1/3}\left(1+2\left(\frac{\delta}{\eta}\right)\right)+6\frac{1+\nu}{1+4\nu}\eta^{2/3}\left(\frac{\delta}{\eta}\right)+\frac{3(1-2\nu)}{1+4\nu}.
\end{align*}
At the center
\begin{align*}
    \kappa^{-1}\prad(\delta_c,\delta_c) = \kappa^{-1}\ptan(\delta_c,\delta_c) &= \frac{3 \left({\delta _c}^{1/3}-1\right)}{1+\nu} \left((2-\nu) {\delta _c}^{1/3}+2 \nu -1\right), \\
    \kappa^{-1}\partial_\delta \prad(\delta_c,\delta_c) &= \frac{3 (1-\nu )}{(\nu +1) \delta _c^{2/3}}.
\end{align*}
The pressures are positive at the center if and only if $\delta_c > 1$ or $0<\delta_c<\Delta_*$, where
\begin{equation*}
    \Delta_* = \left( \frac{1-2\nu}{2-\nu} \right)^3 < 1;
\end{equation*}
 the constant $\Delta_\flat$ is given by $\Delta_\flat = \infty$.
The following theorem was proved in~\cite{AC19}:
\begin{theorem}\label{johntheo}
	When the elastic material is given by the John model, for all $\delta_c:=\rho_c/\mathcal{K}>1$ there exists a unique strongly regular static self-gravitating ball with central density $\rho(0)=\rho_c$. 
\end{theorem}
We are interested in the following question: {\it 
Is $\delta_c > 1$ a necessary condition or can finite radius elastic balls exist in the quasi-linear John material model when \mbox{$0<\delta_c<\Delta_*$}?}

In our numerical investigations we have found evidence suggesting the existence of a constant $\Delta_\circ$, dependent on $\nu$, such that finite radius balls do exist when $\delta_c \in [\Delta_\circ,\Delta_*)$ but not when $\delta_c < \Delta_\circ$, see Figure \ref{fig:John}. We have not a found a closed expression for $\Delta_\circ$, but Figure \ref{fig:John_Delta_circle} shows where in the $(\nu,\delta_c)$-plane finite radius balls could be constructed numerically. An interesting property of these solutions is that the tangential pressure is increasing with the radius rather than decreasing, as it appears to be the case for solutions with center datum $\delta_c > 1$, see Figure \ref{fig:John2}.

\begin{figure}
    \begin{center}
    \subfigure[$\delta_c = 0.90\Delta_*$]{
        \includegraphics[width=0.45\textwidth, trim=0.2cm 0 0.3cm 0,clip]{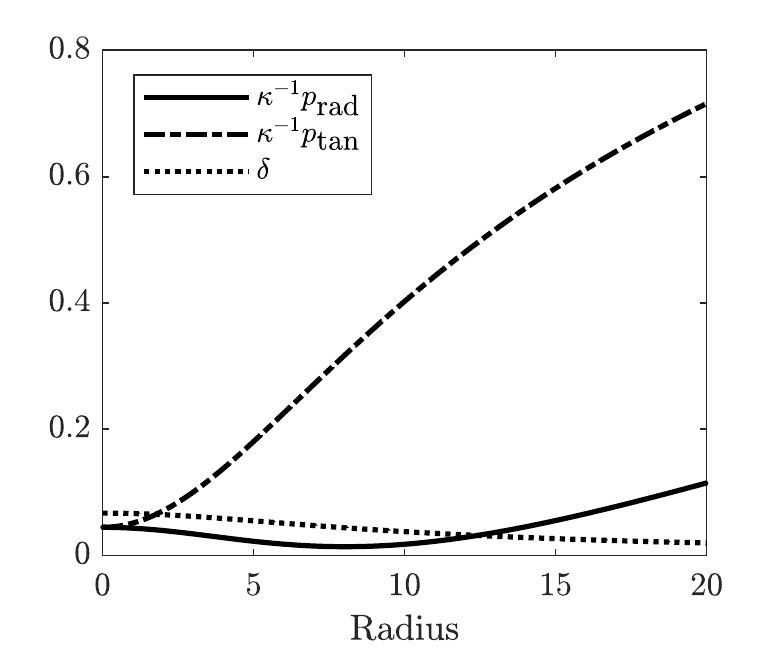}}\quad
        \subfigure[$\delta_c = 0.99\Delta_*$]{
        \includegraphics[width=0.45\textwidth, trim=0.2cm 0 0.3cm 0,clip]{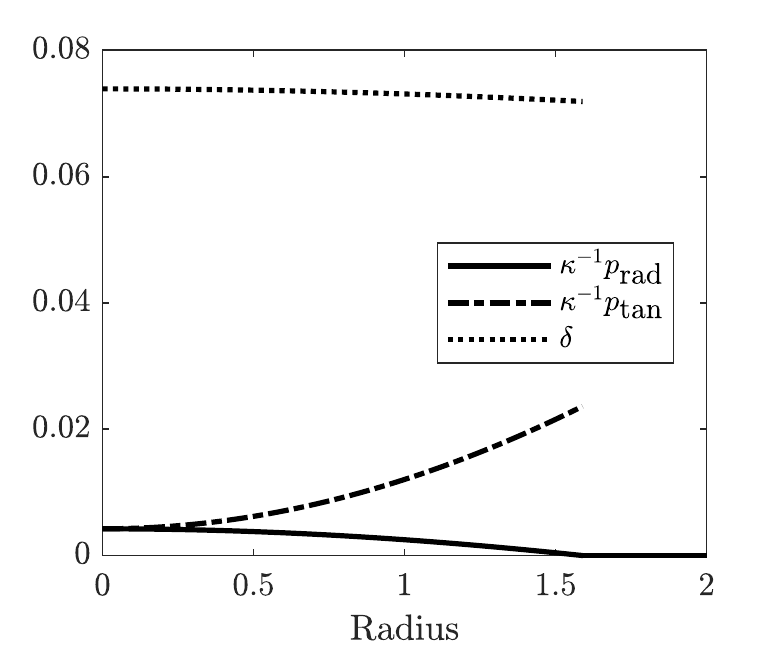}}
    \caption{Elastic balls constructed in the quasi-linear John material model with Poisson ratio $\nu = 0.1$ for center datum smaller than $\Delta_*$ ($\approx 0.0746$). Only in (b) does the ball have finite radius. There seems to exist a $\Delta_\circ$ such that when $\delta_c = \Delta_\circ$ the radial pressure is tangent to the horizontal axis at one point.}
    \label{fig:John}
    \end{center}
\end{figure}

\begin{figure}
\begin{center}
\includegraphics[width=0.48\textwidth,trim=0cm 0 0.5cm 0, clip ]{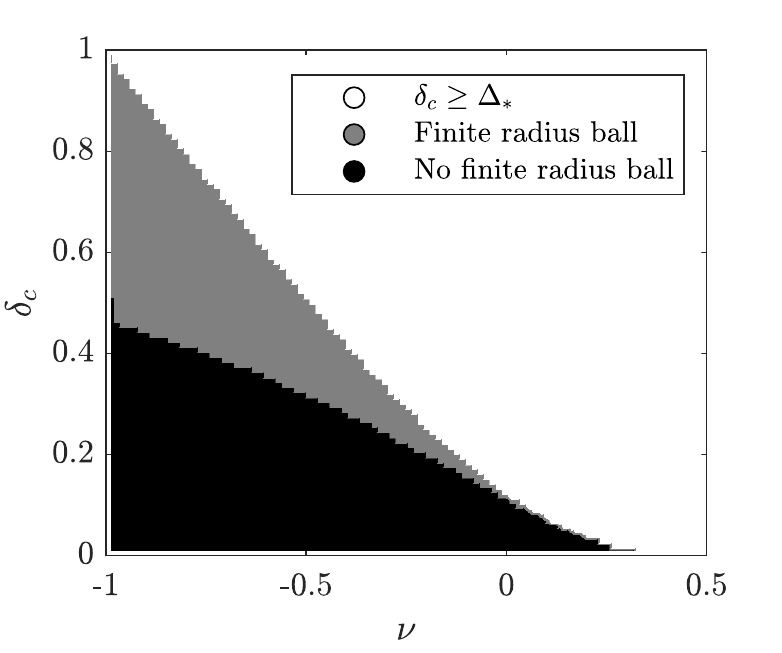}
\caption{The gray region indicates where finite radius balls have been found numerically in the quasi-linear John model for $\delta_c \in (0,\Delta_*)$. The border between the black and gray regions approximates the proposed $\Delta_\circ(\nu)$.}
    \label{fig:John_Delta_circle}
\end{center}
\end{figure}


\begin{figure}
    \begin{center}
    \subfigure[$\nu=-0.5$]{
        \includegraphics[width=0.45\textwidth, trim=0.2cm 0 0.3cm 0,clip]{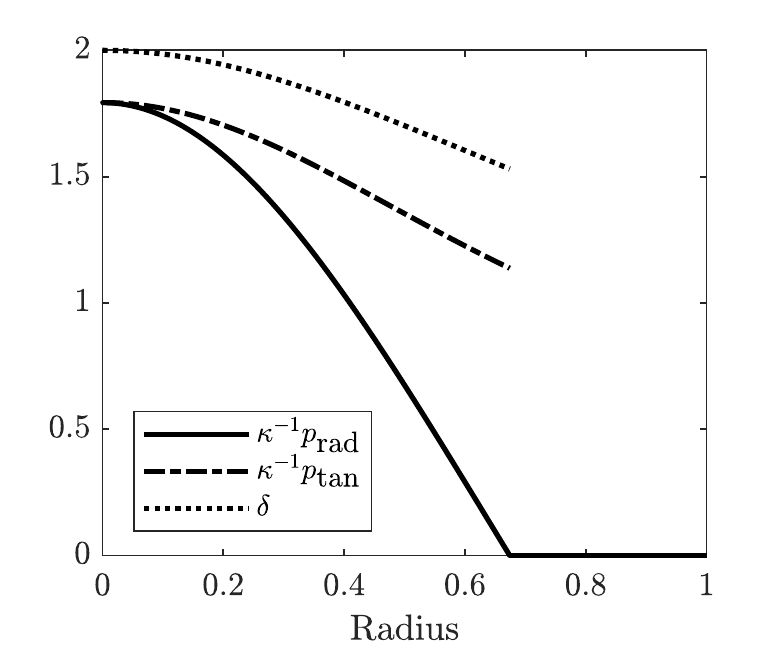}}\quad
        \subfigure[$\nu=0.25$]{
        \includegraphics[width=0.45\textwidth, trim=0.2cm 0 0.3cm 0,clip]{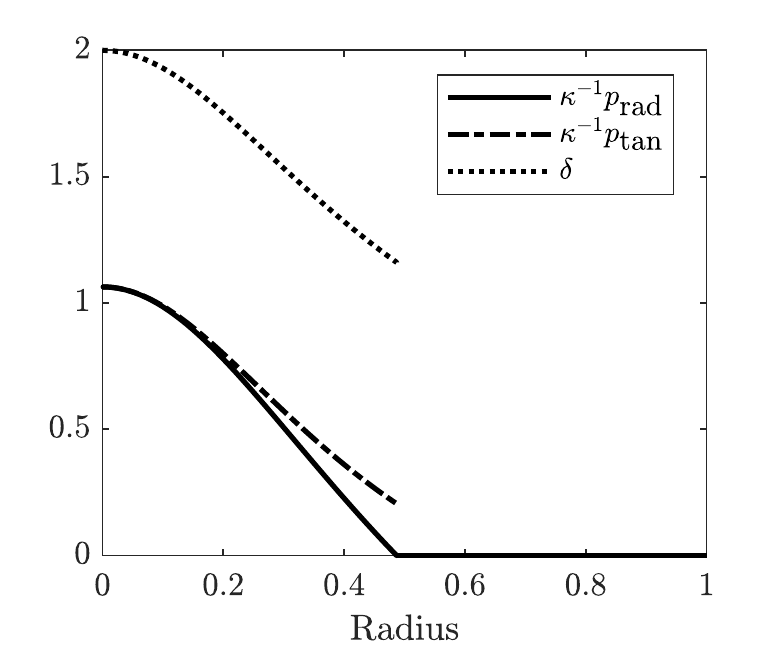}}
    \caption{Elastic balls constructed in the quasi-linear John material model with center datum $\delta_c = 2$. The tangential pressure decreases with the radius as opposed to when $\delta_c < 1$, see Figure \ref{fig:John}.}
    \label{fig:John2}
    \end{center}
\end{figure}

    \subsection{Hadamard materials}\label{sec:hadamard}
Hadamard materials are hyperelastic materials with stored energy function~\eqref{Had},
which yields the principal pressures
\begin{align*}
&\prad(\delta,\eta)=-\frac{3}{1+\nu}\eta^{-1/3}\left(\frac{\delta}{\eta}\right)^{-1}+\frac{3\nu}{1+\nu}\eta^{1/3}\left(\frac{\delta}{\eta}\right)^{-1}+3\frac{1-\nu}{1+\nu},\\
&\ptan(\delta,\eta)=-\frac{3}{2(1+\nu)}\eta^{-1/3}\left(\left(\frac{\delta}{\eta}\right)^{-1}+\left(\frac{\delta}{\eta}\right)\right)+\frac{3\nu}{1+\nu}\eta^{1/3}\left(\frac{\delta}{\eta}\right)+3\frac{1-\nu}{1+\nu}
\end{align*}
and at the center we have
\begin{align*}
    \kappa^{-1}\prad(\delta_c,\delta_c) = \kappa^{-1}\prad(\delta_c,\delta_c) &= \frac{3 \left({\delta _c}^{1/3}-1\right) \left(\nu  {\delta _c}^{1/3}+1\right)}{(\nu +1) {\delta _c}^{1/3}}, \\
    \kappa^{-1}\partial_\delta \prad(\delta_c,\delta_c) &= \frac{3 \left(1-\nu  \delta _c^{2/3}\right)}{(\nu +1) \delta _c^{4/3}}.
\end{align*}
The pressures are positive at the center if and only if $\delta_c>1$; the constant $\Delta_\flat$ is given by
\begin{equation*}
    \Delta_\flat=\left(\frac{1}{\nu}\right)^{3/2}.
\end{equation*}
The following theorem has been proved in~\cite{AC19}:
\begin{theorem}\label{hadtheo}
When the elastic material is given by the Hadamard model, the condition $\delta_c:=\rho_c/ \mathcal{K}>1$ is necessary for the existence of regular static self-gravitating balls. For 
\[
1<\delta_c<\left(\frac{1}{2\nu}\right)^{3/2}=\Delta_\sharp
\] 
there exists a unique strongly regular static self-gravitating ball with central density $\rho(0)=\rho_c$. 
\end{theorem}
We are interested in the following question: {\it
Is the sufficient bound $\delta_c < \Delta_\sharp$ in theorem \ref{hadtheo} necessary or can it be replaced by the weaker bound $\delta_c < \Delta_\flat$?}

In our numerical investigations we have found evidence suggesting that the bound indeed can be replaced. Figure \ref{fig:Hadamard_sharp} shows finite radius balls with center datum both smaller than and larger than $\Delta_\sharp$.
Furthermore, finite radius solutions seem to exist for center datum in the entire interval from $\Delta_\sharp$ up to $\Delta_\flat$, but not for $\delta_c \geq \Delta_\flat$, see Figure \ref{fig:Hadamard_flat}. When $\delta_c \geq \Delta_\flat$, the hyperbolicity condition is violated at the center and the numerical solutions in Figure \ref{fig:Hadamard_flat}(b) blows up similarly to what they do in the Saint Venant-Kirchhoff model. 


\begin{figure}
    \begin{center}
    \subfigure[$\delta_c = 0.99\Delta_\sharp$]{
        \includegraphics[width=0.45\textwidth, trim=0.2cm 0 0.3cm 0,clip]{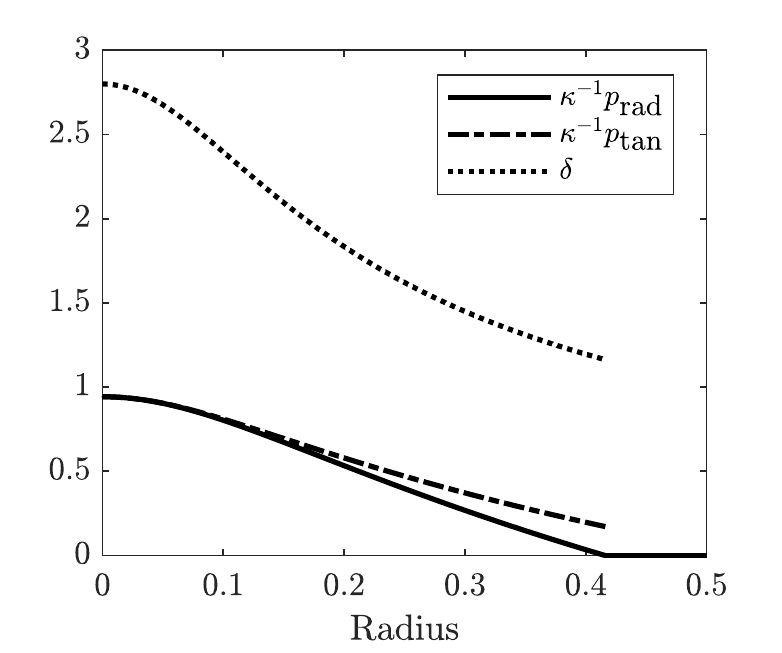}}\quad
        \subfigure[$\delta_c = 1.01\Delta_\sharp$]{
        \includegraphics[width=0.45\textwidth, trim=0.2cm 0 0.3cm 0,clip]{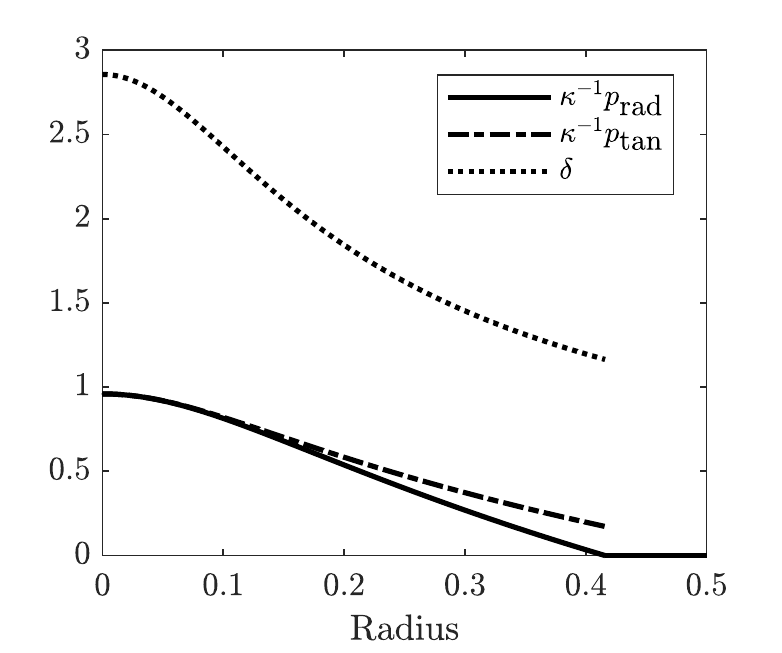}}
    \caption{Elastic balls constructed in the Hadamard material model with Poisson ratio $\nu = 0.25$ for center datum close to $\Delta_\sharp$ ($\approx 2.83$).}
    \label{fig:Hadamard_sharp}
    \end{center}
\end{figure}


\begin{figure}
    \begin{center}
    \subfigure[$\delta_c = 0.99\Delta_\flat$]{
        \includegraphics[width=0.45\textwidth, trim=0.2cm 0 0.3cm 0,clip]{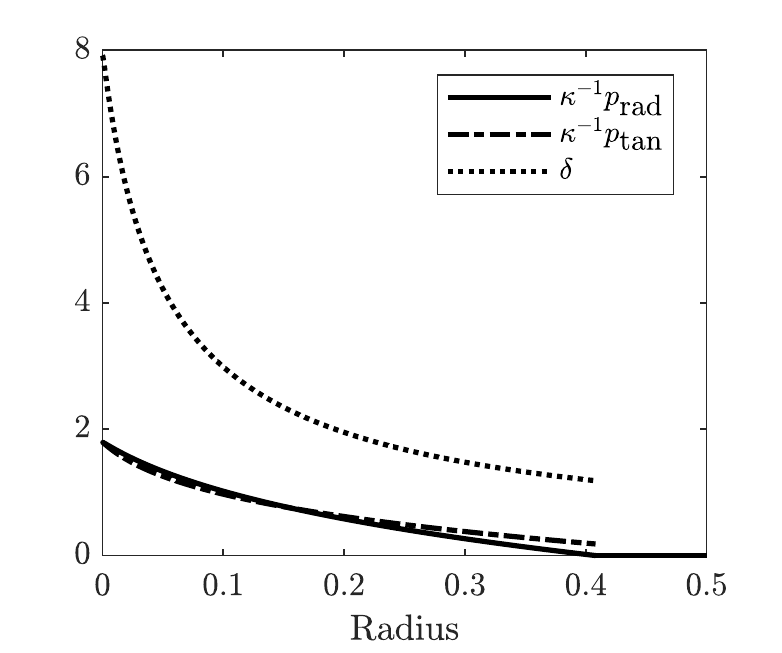}}\quad
        \subfigure[$\delta_c = 1.01\Delta_\flat$]{
        \includegraphics[width=0.45\textwidth, trim=0.2cm 0 0.3cm 0,clip]{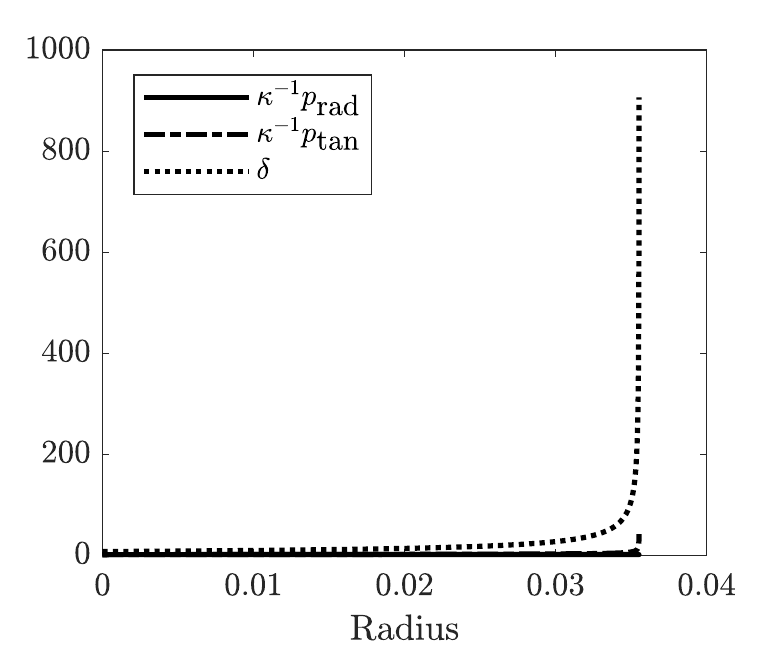}}
    \caption{Elastic balls constructed in the Hadamard material model with Poisson ratio $\nu = 0.25$ for center datum close to $\Delta_\flat$ ($\approx8$).}
    \label{fig:Hadamard_flat}
    \end{center}
\end{figure}

\subsection{Quasi-linear Signorini materials}\label{sec:signorini}
Quasi-linear Signorini materials are hyperelastic with stored energy function~\eqref{Sig}$_{\tau=0}$,
which yields the principal pressures
\begin{align*}
&\prad(\delta,\eta)\\
&\quad=\frac{3(1+4\nu)}{8(1+\nu)}\eta^{2/3}\left(2-\left(\frac{\delta}{\eta}\right)^2\right)+\frac{3}{16(1+\nu)}\eta^{4/3}\left(-4+4\left(\frac{\delta}{\eta}\right)^2+3\left(\frac{\delta}{\eta}\right)^4\right)-\frac{3(5+8\nu)}{16(1+\nu)}\\
&\ptan(\delta,\eta)=\frac{3(1+4\nu)}{8(1+\nu)}\eta^{2/3}\left(\frac{\delta}{\eta}\right)^2-\frac{3}{16(1+\nu)}\eta^{4/3}\left(-4+\left(\frac{\delta}{\eta}\right)^4\right)-\frac{3(5+8\nu)}{16(1+\nu)}.
\end{align*}
At the center
\begin{align*}
    \kappa^{-1}\prad(\delta_c,\delta_c) = \kappa^{-1}\ptan(\delta_c,\delta_c) &= \frac{3 \left(\delta _c^{2/3}-1\right) \left(3 \delta _c^{2/3}+8 \nu +5\right)}{16 (\nu +1)},\\
    \kappa^{-1}\partial_\delta \prad(\delta_c,\delta_c) &= \frac{3 \left(5 \delta _c^{2/3}-4 \nu -1\right)}{4 (\nu +1) {\delta _c}^{1/3}}.
\end{align*}
The pressures are positive at the center if and only if $\delta_c > 1$ or, when $\nu \in (-1, -5/8)$, for $0 < \delta_c < \Delta_*$, where
\begin{equation*}
    \Delta_* = \left( \frac{-5-8\nu}{3} \right)^{3/2} < 1.
\end{equation*}
Furthermore, the constant $\Delta_\flat$ is given by $\Delta_\flat = \infty$. 

{\it Remark.} The hyperbolicity condition is violated at the center if $\delta_c < ((4\nu+1)/5)^{3/2}$ but this only happens for invalid combinations of $\delta_c$ and $\nu$, i.e., when the principal pressures are negative at the center.

The results in \cite{AC19} do not cover the quasi-linear Signorini materials, so we are interested in the following questions: {\it 
Do finite radius elastic balls exist in the quasi-linear Signorini material model when (a) $\delta_c>1$, (b) $0<\delta_c<\Delta_*$?}

Regarding question (a) we were able to numerically construct finite radius balls for every combination of $\nu \in (-1,1/2)$ and $\delta_c>1$ that we tried, see Figure \ref{fig:Signorini2} for examples. Regarding question (b), we found numerical evidence suggesting the existence of a constant $\Delta_\circ$, dependent on $\nu \in (-1, -5/8)$, such that finite radius balls do exist when $\delta_c \in [\Delta_\circ,\Delta_*)$ but not when $\delta_c < \Delta_\circ$, see Figure \ref{fig:Signorini}. We have not been able to derive a closed expression for $\Delta_\circ$, but Figure \ref{fig:Signorini_Delta_circle} shows where in the $(\nu,\delta_c)$-plane finite radius balls could be constructed numerically.

\begin{figure}
    \begin{center}
    \subfigure[$\nu=-0.5$]{
        \includegraphics[width=0.45\textwidth, trim=0.2cm 0 0.3cm 0,clip]{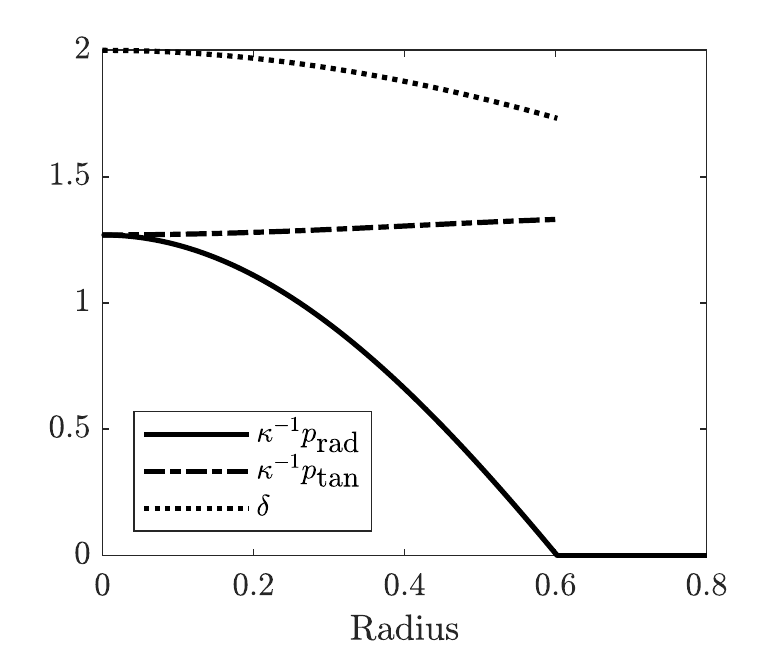}}\quad
        \subfigure[$\nu=0.25$]{
        \includegraphics[width=0.45\textwidth, trim=0.2cm 0 0.3cm 0,clip]{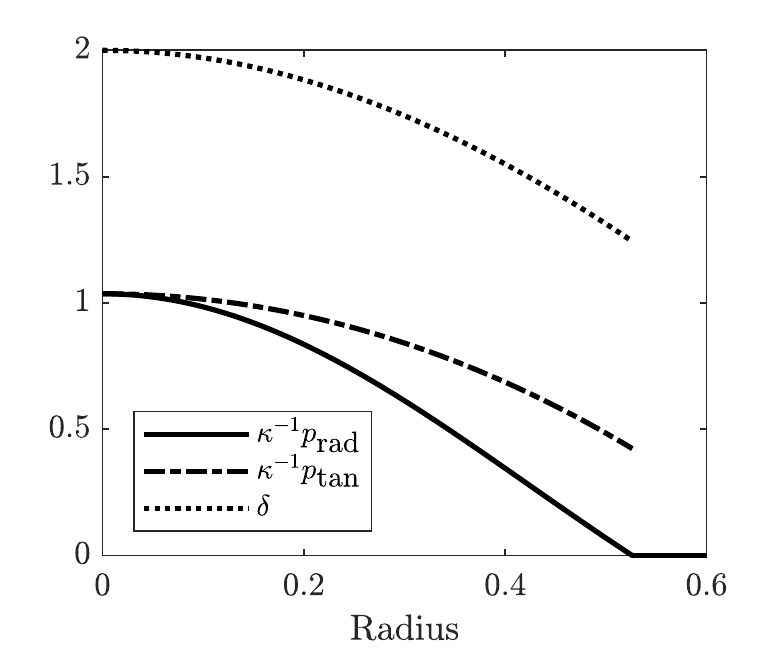}}
    \caption{Elastic balls constructed in the quasi-linear Signorini material model with center datum $\delta_c = 2$.}
    \label{fig:Signorini2}
    \end{center}
\end{figure}
%

\begin{figure}
    \begin{center}
    \subfigure[$\delta_c = 0.65\Delta_*$]{
        \includegraphics[width=0.45\textwidth, trim=0.2cm 0 0.3cm 0,clip]{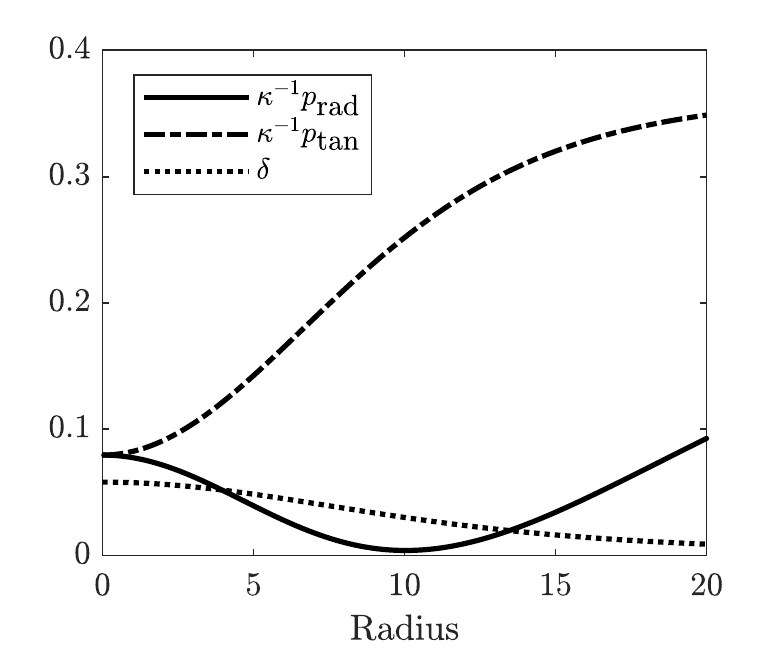}}\quad
        \subfigure[$\delta_c = 0.70\Delta_*$]{
        \includegraphics[width=0.45\textwidth, trim=0.2cm 0 0.3cm 0,clip]{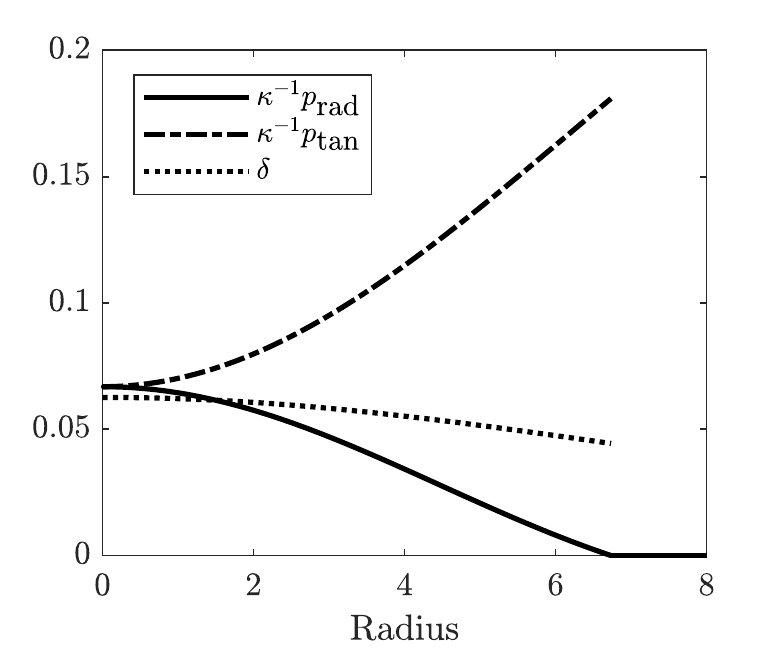}}
    \caption{Elastic balls constructed in the quasi-linear Signorini material model with Poisson ratio $\nu = -0.7$ for center datum smaller than $\Delta_*$ ($\approx 0.0894$). Only in (b) does the ball have finite radius. There seems to exist a $\Delta_\circ$ such that when $\delta_c = \Delta_\circ$ the radial pressure is tangent to the horizontal axis at one point.}
    \label{fig:Signorini}
    \end{center}
\end{figure}

\begin{figure}
\begin{center}
\includegraphics[width=0.48\textwidth,trim=0cm 0 0.5cm 0, clip ]{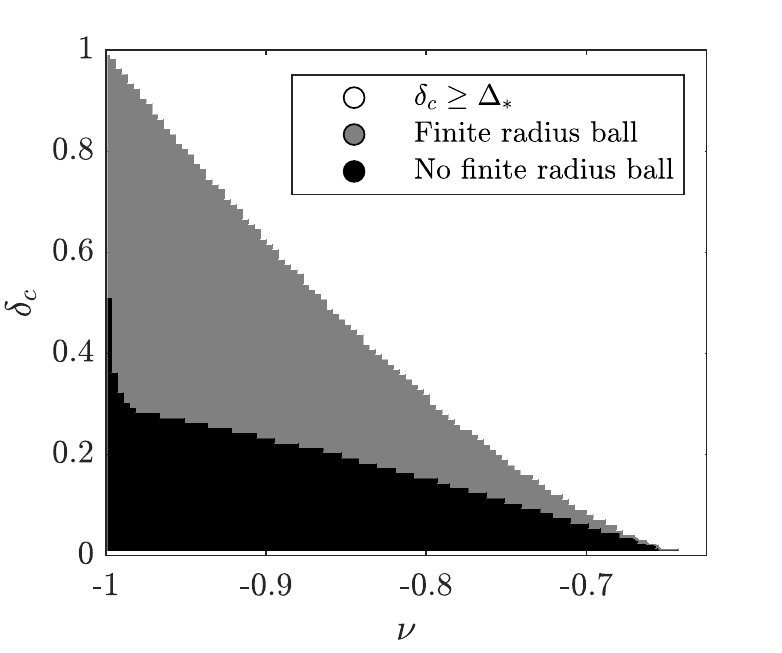}
\caption{The gray region indicates where finite radius balls have been found numerically in the quasi-linear Signorini model for $\delta_c \in (0, \Delta_*)$. The border between the black and gray regions approximates the proposed $\Delta_\circ(\nu)$}
    \label{fig:Signorini_Delta_circle}
\end{center}
\end{figure}


\newpage

\begin{thebibliography}{10}

\bibitem{AC18} A. Alho, S. Calogero: Multi-body spherically symmetric steady states of Newtonian self-gravitating elastic matter.
Comm. Math. Phys. {\bf 371}, 975--1004 (2019)

%
\bibitem{AC19} A. Alho, S. Calogero: Static self-gravitating Newtonian elastic balls. Archiv. Rat. Mech. Anal. {\bf 238}, 639--669 (2020)



\bibitem{BT} J.~Binney, S.~Tremaine: {\it Galactic Dynamics}. Princeton Series in Astrophysics (1987)

\bibitem{BS1} R.~Beig, B.~G.~Schmidt: Static, Self-Gravitating Elastic Bodies. Proc.~Roy.~Soc.~Lond. A {\bf 459}, 109--115 (2003) 

\bibitem{CT} S.~Calogero, T.~Leonori: Ground states of self-gravitating elastic bodies. 
Calc. Var. and PDE {\bf 54}, 881-899 (2015)

\bibitem{SC} S.~Calogero: On self-gravitating politropic elastic balls. Preprint arXiv:2104.11126 (2021)

\bibitem{CH} N.~Chamel, P.~Haensel: Physics of Neutron Star Crusts. Living Review in Relativity {\bf 11}, 10 (2008)

\bibitem{Ciarlet} P.~G.~Ciarlet: {\it Mathematical elasticity, Vol I: Three dimensional elasticity}. North-Holland (1988)
%
%
%
%
%

\bibitem{Jeans}J.~H.~Jeans: On the vibrations and stability of a gravitating planet. Phil. Trans. R. Soc. Lond. A {\bf 201}, 331--345 (1903)

\bibitem{JKJG} F.~Jia, O.~Kodio, S.~J.~Chapman, A.~Goriely: On the figure of elastic planets I: gravitational collapse and infinitely many equilibria. Proc. Royal Soc. A {\bf 475}: 20180815 (2019)

\bibitem{John} F.~John: Plane strain problems for a perfectly elastic material of harmonic type. Comm. Pure Appl. Math.
 {\bf 13}, 239--260 (1960)
 
\bibitem{KW} R.~Kippenhahn, A.~Weigert, A.~Weiss: {\it Stellar Structure and Evolution}. Springer-Verlag
Berlin Heidelberg (2012)

\bibitem{Astrid} A.~Liljenberg: {\it Spherically symmetric self-gravitating elastic bodies: A numerical investigation}. Master Thesis, Chalmers University of Technology, Gothenburg (2020)

 \bibitem{Lord} O.~M.~Lord~Rayleigh. On the dilatational stability of the Earth. Proceedings of the Royal Society of London {\bf 77}, 486--499 (1906)
 
 \bibitem{Love} A.~E.~H.~Love: {\it A treatise on the mathematical theory of elasticity}. Cambridge University Press (1892)
 
\bibitem{Love2}   A.~E.~H.~Love: {\it Some problems of geodynamics being an essay to which the Adams prize in the
University of Cambridge was adjusted in 1911}. Cambridge University Press (1911)

\bibitem{MH} J.~E.~Marsden, T.~J.~R.~Hughes: {\it Mathematical foundations of elasticity}. Dover publications, New York (1994)


\bibitem{MW16} W.~Mueller, W. Weiss: {\it The State of Deformation in Earthlike Self-Gravitating Objects}. 
SpringerBriefs in Applied Sciences and Technology - continuum mechanics (2016)

\bibitem{Ogden} R.~W.~Ogden: {\it Non-linear elastic deformations.} Dover Publications, New York (1984)

\bibitem{Sig43} A.~Signorini: Trasformazioni termoelastiche finite, Memoria 1. Annali di Matematica Pura ed Applicata {\bf 22}, 33--143 (1943) 
\end{thebibliography}
\end{document}